\documentclass[letterpaper, 10 pt, conference]{ieeeconf}  

\IEEEoverridecommandlockouts                              
\usepackage[margin=0.75in]{geometry}
\usepackage{amsmath} 
\usepackage{amssymb} 

\newtheorem{assumption}{Assumption}
\newtheorem{theorem}{Theorem}

\usepackage{graphicx}
\graphicspath{{./images/}}
\usepackage{caption}
\usepackage{subcaption}
\usepackage{xcolor}
\definecolor{kthblue}{rgb}{0.098, 0.329, 0.651}
\usepackage{biblatex}
\usepackage[final]{changes}
\usepackage{algpseudocode}
\usepackage{subcaption}
\usepackage{algorithm}
\usepackage{bbm}

\definecolor{newred}{RGB}{153, 0, 0}
\begin{document}
\title{\LARGE \bf
Model-free fast charging of lithium-ion batteries\\
by online gradient descent
}
\author{Hamed Taghavian, Malin Andersson, and Mikael Johansson
\thanks{This work was supported by KTH Royal Institute of Technology.
The authors are with the Division of Decision and Control Systems, KTH Royal Institute of Technology. Emails: \tt\small{hamedta@kth.se}, \tt\small{malinan9@kth.se}, \tt\small{mikaelj@kth.se}.}}

\maketitle
\thispagestyle{empty}
\pagestyle{empty}

\begin{abstract}
A data-driven solution is provided for the fast-charging problem of lithium-ion batteries with \replaced{multiple safety and aging}{various} constraints. 
 The proposed method optimizes the charging current based on the observed history of measurable battery quantities, such as the input current, terminal voltage, and temperature.  The proposed method does not need any detailed battery model or full-charging training episodes. The theoretical convergence is proven under mild conditions and is validated numerically on several linear and nonlinear battery models, including single-particle and equivalent-circuit models.
\end{abstract}

\section{INTRODUCTION}
Whether it be trains, road vehicles, or boats, electrified transportation is believed to have a positive impact on both the ecology and the economy. Among the different energy storage units available for electrified transportation, lithium-ion batteries are expected to remain the dominant option for at least a decade, due to their many favorable properties~\cite{SWHPW2018}.

Despite their numerous advantages, however, electric vehicles equipped with lithium-ion batteries suffer from long charging times compared to conventional gasoline-fueled vehicles~\cite{NaS2021}. Simply increasing the charging current does not solve this problem due to the safety and degradation issues associated with high charging currents. Hence, a recent line of research has focused on developing charging protocols that 
maximize the charging current while accounting for battery health and safety~\cite{TCFOLCJELLo2019}.

Obtaining such optimal charging protocols normally requires accurate models of the battery dynamics which are not always available. Identifying the parameters that determine the dynamic behavior of lithium-ion batteries is difficult~\cite{LDCJRJS2022} and there are several factors, such as aging, ambient conditions, etc., that can cause the actual dynamics to drift away from those of a previously identified model.
%
Therefore, several model-free charging techniques have been proposed that avoid system identification and instead operate directly based on observing the data collected from the lithium-ion batteries in real-world operation. Available approaches include Bayesian optimization~\cite{JiW2021}, iterative learning control~\cite{WaXL2019} and reinforcement learning~\cite{HLWJ2023}. To converge to the optimal charging protocol, these methods rely on a learning process over several full-charging episodes. Such charging experiments are time-consuming and degrade the batteries.

In this paper, we first note that the optimal solutions to typical fast-charging problems are of the bang-ride form, that is, they operate on the boundary of the feasible set~\cite{DCHCG2023}. This property has already been proven for various fast-charging problems~\cite{selector,malinIFAC,PLATM2020} and conjectured to be true for several others~\cite{LVWWZTWZ2021}. The industry standard for fast charging of lithium-ion batteries, the Constant-Current-Constant-Voltage (CCCV) charging protocol, is of bang-ride form, with an initial phase where the maximal current is applied (bang), followed by a phase where the current is gradually decreased to maintain a constant terminal voltage (ride). We exploit this fact to design a data-driven controller with a fixed structure that converges to the bang-ride charging protocols in real-time, without requiring any prior training episodes. The proposed approach is less costly than other model-free methods, as it does not need additional charging experiments.

The proposed controller uses measured input-output data of the battery to learn its bang-ride protocol, without relying on a battery model or access to the cost function in the fast-charging problem. The protocol does not need to know the functional form of the constraints, it only assumes that constraint violations can be measured. 
Our proposed control structure has a similar switching nature to \cite{LVWWZTWZ2021}. However, the approach in~\cite{LVWWZTWZ2021} requires a full model and its parameters and does not provide a technique for tuning the PI controllers nor proves convergence to the desired charging protocol. In contrast, we propose an adaptive feedback controller that guarantees convergence to the bang-ride protocol without a battery model or manual tuning of the controllers.

Convergence to the bang-ride protocol is realized by solving an online convex optimization problem. Online convex optimization is an adaptation of convex optimization algorithms to unknown and changing environments~\cite{Zink2003}. In these problems, the cost function changes in every time step, and the algorithm only has access to the previous set of cost functions and decision variables~\cite{hazan2016}. An additional complexity emerges when the cost functions in different time steps are coupled~\cite{NoM2022,SiDPG2020}. 
This happens in optimal control problems~\cite{NoM2022} as well as in fast-charging problems since there is an underlying dynamic system that affects the cost functions in the overall optimization process. Nevertheless, the online optimization-based algorithm proposed in this paper can achieve a sublinear regret
similar to the uncoupled case~\cite{Zink2003}, and contrary to~\cite{NoM2022}, it does not require linearity of the dynamics and strong convexity of the cost functions. Furthermore, in contrast with other learning-based control systems~\cite{HoW2013}, such as iterative feedback tuning (IFT) \cite{ift} and simultaneous perturbation stochastic approximation (SPSA) \cite{spall}, the monotone properties specific to fast charging problems are exploited to avoid estimating the gradient directions. Thus, the charging current is not disrupted by running experiments during the learning process. 

This paper is structured as follows. We define the fast-charging optimization problems in Section~\ref{sec:problem_statement}, develop a model-free fast-charging algorithm based on online gradient descent in Section~\ref{sec:online}, and present a few case studies for numerical evaluation in Section~\ref{sec:numeric}. Conclusive remarks and potential future research directions are provided in Section~\ref{sec:conclusion}.

\section{FAST CHARGING AND BANG-RIDE PROTOCOLS}\label{sec:problem_statement}
Consider a discrete-time battery model of the form 
\begin{subequations}\label{eqn:system}
\begin{align}
    x_{t+1}=f(x_t,u_t), &\quad  t\in\mathbb{N}_0 \label{eqn:system(a)}\\
    y_{i,t}=h_i(x_t,u_t), &\quad  i=1,2,\dots, p \label{eqn:system(b)}
\end{align}
\end{subequations}
where $x_t\in\mathbb{R}^n$ is the battery state, $y_{i,t}\in\mathbb{R}$ are the measurable outputs and $u_t\in\mathbb{R}$ is the input current. As a convention,  $u_t$ is positive when we charge (feed current into the battery). A general fast-charging optimization problem with multiple constraints can be stated as
\begin{equation}
    \begin{array}[c]{rll}
    \underset{u}{\text{maximize}} & J(x_0,u)=\sum_{t=0}^{t_f}L(x_t,u_t) \\
    \mbox{subject to} & y_{i,t}\leq \Bar{y}_i,\quad i=1,2,\dots,p\\
    & x,u,y_i \textnormal{ satisfy (\ref{eqn:system})} \\
    \end{array} \label{eqn:optimization}
\end{equation}
where $x_0\in\mathbb{R}^n$ is the initial state of the battery. We make the following assumption throughout the paper.
\begin{assumption}\label{ass:h}
    The output functions $h_i(x_t,u_t)$ are differentiable and strictly increasing in $u_t\in \mathbb{R}$ for all $x_t\in \mathbb{R}^n$.
\end{assumption}

We assume the first constraint in (\ref{eqn:optimization}) imposes an explicit upper bound on the charging current, \emph{i.e.},
\begin{equation}\label{eqn:u<umax}
    y_{1,t}=h_1(x_t,u_t)=u_t \leq u_{\rm max}=\bar{y}_1
\end{equation}

In typical fast-charging problems, the running-cost $L(x_t,u_t)$ and transition functions  $f(x_t,u_t)$ are continuous and increasing in $x_t,u_t$, and the constraint functions $h_i(x_t,u_t)$ are continuous and increasing in $u_t$~\cite{DCHCG2023}. This makes (\ref{eqn:optimization}) a monotone optimal control problem, whose solution is known to touch the feasible set boundaries, regardless of the exact transition, cost, or constraint functions~\cite{selector}. If the optimal solution operates on the feasible set boundary at all times, \emph{i.e.}, at least one inequality constraint in (\ref{eqn:optimization}) is active at every time step, then the optimal solution is said to be of bang-ride form. 
This has shown to be the case for several specific fast-charging problems on the form (\ref{eqn:optimization})~\cite{malinIFAC,PLATM2020}. More general sufficient conditions that ensure the optimal solution of (\ref{eqn:optimization}) is of bang-ride form are provided in~\cite{selector}. Hence in this paper, we provide an algorithm that implements the bang-ride solution of the fast-charging problem (\ref{eqn:optimization}) when the battery model, the constraints, and the cost function are unknown.

To calculate the bang-ride solution of (\ref{eqn:optimization}), we note that since the output functions are strictly increasing (Assumption~\ref{ass:h}), each equation
\begin{equation}\label{eqn:constraint_equation}
    h_i(x_t,u_t)=\Bar{y}_i, \quad i=1,2,\dots,p
\end{equation}
has at most one solution, $u_t=K_i(x_t)$ where $K_i:\mathbb{R}^n\to \mathbb{R}$ is a possibly nonlinear function. For example, for the first constraint (\ref{eqn:u<umax}), equation (\ref{eqn:constraint_equation}) has the uninque solution
$$
u_t=K_1(x_t)=u_{\rm max}
$$
By defining $K_i(x_t)=+\infty$ in case equation (\ref{eqn:constraint_equation}) has no real solution for $u_t$, it is possible to write the ideal bang-ride protocol as the following nonlinear state-feedback law
\begin{equation}\label{eqn:selector}
    u_t=\min\lbrace K_1(x_t),
    K_2(x_t),\dots, K_p(x_t)\rbrace
\end{equation}
This expression is well-defined since at least one argument, \emph{i.e.} $K_1(x_t)$, is always finite.


If an accurate model is available and all system states can be measured, the bang-ride control (\ref{eqn:selector}) is simple to implement. Even if the feedback functions $K_i$ cannot be determined explicitly, the monotonicity of the output functions makes it easy to solve the constraint equations (\ref{eqn:constraint_equation}) numerically for a given $x_t$ to determine each $K_i(x_t)$. 

In practice, however, battery models are not perfect and the full system state is not measured. It is then considered to be difficult to obtain a high-quality estimate of the model parameters and the state vector $x_t$~\cite{LDCJRJS2022}. To overcome these challenges, we develop a data-driven implementation of (\ref{eqn:selector}) in the next section that does not rely on the system model, does not require the state vector $x_t$, and  does not require solving the constraint equations (\ref{eqn:constraint_equation}) in every sample.  

\section{MODEL-FREE CHARGING}\label{sec:online}

\subsection{Active constraint switching}
The bang-ride control law (\ref{eqn:selector}) effectively determines both which constraint should be active and the control input that keeps the corresponding output at its maximally allowed value. When we do not have access to an accurate model and cannot measure the full state vector, we need other mechanisms to determine what constraints to activate and how the input should be adjusted to keep the corresponding output function at its limit. These mechanisms should only rely on the observed system outputs $y_{i,t}$.

Let us define the \emph{constraint errors} which measure deviations from the inequality constraint boundaries in (\ref{eqn:optimization}) as follows
\begin{equation}\label{eqn:cnstr_error}
    e_{i,t}=\gamma_i\bigl(\Bar{y}_i - y_{i,t}\bigr), \quad i=1,2,\dots,p
\end{equation}
where $\gamma_i$'s are some positive constants. The $i$-th inequality constraint in (\ref{eqn:optimization}) is satisfied at time $t$ if and only if $e_{i,t}\geq 0$, and it is said to be active when $e_{i,t}= 0$. The active constraint, \emph{i.e.} the one that the bang-ride controller aims to ride, is then 
\begin{equation}\label{eqn:i_star}
    i^{\star}(t)=\arg\min_i \lbrace e_{i,t}\rbrace
\end{equation} 
Equation (\ref{eqn:i_star}) yields in a switching mechanism based on the constraint errors (\ref{eqn:cnstr_error}) which, in constrast to (\ref{eqn:selector}), are directly measurable. For notational convenience, we sometimes 
drop the argument $t$ and write $i^{\star}=i^{\star}(t)$.

\subsection{Active constraint riding}
Once the constraint index $i^{\star}$ is determined, an appropriate control signal $u_t$ is needed to bring the respective error in (\ref{eqn:cnstr_error}) to zero so that the selected constraint is ridden. 
This problem is equivalent to designing an output-feedback controller for the unknown plant  (\ref{eqn:system}) to track the reference signal $r_t=\bar{y}_i$. 

\begin{figure}
        \centering
        \includegraphics[width=1\linewidth]{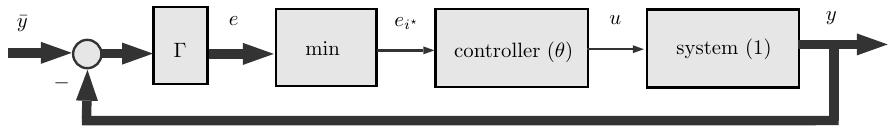}
        \caption{The data-driven control scheme: The thick lines represent vector signals in $\mathbb{R}^p$ and $\Gamma=\textnormal{diag}(\gamma)$ is a non-negative diagonal matrix used to adjust the loop gains if necessary. The loop is only connected in one component at a time, which corresponds to the active constraint.}
        \label{fig:scheme}
\end{figure}

For this purpose, we consider the control structure in Figure~\ref{fig:scheme}, where the controller block uses a fixed differentiable structure $u_t=C(\mathcal{H}_t,\theta_t)$ parametrized by $\theta_t \in\mathbb{R}^m$, where
\begin{equation}\label{eqn:history}
    \mathcal{H}_t=\left\lbrace 
    \left.
    \left(u_k,y_{i^{\star}(k),k}\right)
    \right\vert 
    0\leq k\leq t-1\right\rbrace
\end{equation}
is the observed input-output data history just before $t$. Since an accurate model of the system dynamics is not available, 
we propose a learning-based control law whose parameters $\theta$ are adapted based on the observed data. The underlying algorithm that tunes the parameters attempts to minimize the one-step-ahead constraint error squared (\ref{eqn:cnstr_error}) at every time step, by solving the following online optimization problem
\begin{equation}
    \begin{array}[c]{rll}
    \underset{\theta\in\Theta}{\text{minimize}} & J_{t}(\theta)=e_{i^{\star},t}^2\\
    \mbox{subject to}
    & u_t=C(\mathcal{H}_t,\theta)\\
    & u,e_i \textnormal{ satisfy (\ref{eqn:system}), (\ref{eqn:cnstr_error})} \\
    \end{array} \label{eqn:online_opt}
\end{equation}
where $\Theta\subseteq\mathbb{R}^m$. The gradient of $J_{t}(\theta)$ at $\theta=\theta_{t}$ is given by
\begin{equation}\label{eqn:true_gradient}
    \nabla_{\theta}J_{t}(\theta_t)=-c_te_{i^{\star},t}\nabla_\theta C(\mathcal{H}_t,\theta_t)
\end{equation}
where
\begin{equation}\label{eqn:c_t}
    c_t=2\gamma_{i^{\star}}\frac{\partial h_{i^{\star}}}{\partial u}(x_t,u_t) 
\end{equation}
Since the output functions are monotonic by Assumption~\ref{ass:h} and $\gamma_i>0$,  $c_t$ is non-negative. Thus, the vector
\begin{equation}\label{eqn:g}
    g_t=-e_{i^{\star},t}\nabla_\theta C(\mathcal{H}_t,\theta_t)
\end{equation}
points to the exact gradient direction and 
is computable solely based on the observed data. Therefore, we use (\ref{eqn:g}) to perform one projected gradient descent step, \emph{i.e.},
\begin{equation}\label{eqn:gradient_step}
    \theta_{t+1}=\operatorname{proj}_{\Theta} \lbrace \theta_{t}-\alpha_{t} g_t \rbrace
\end{equation}
towards the optimal parameter $\theta=\theta^{\star}_t$ that solves (\ref{eqn:online_opt}). 
The proposed method is summarized in Algorithm~\ref{Alg:BangRide}.

\begin{algorithm}
\caption{Data-driven bang-ride control:}\label{alg}
\begin{algorithmic}[1]
\Require $t_f,\Bar{y}_i$ ($i=1,2,\dots,p$) \Comment{input data.}
\Require $\theta_{0}\in\mathbb{R}^m$ \Comment{initial parameters.}
\State $t= 0$, $\mathcal{H}_{0}= \emptyset$, $i^{\star}=1$ 
      \While{$t\leq t_f$}
        \State $u_t \gets C(\mathcal{H}_t,\theta_t)$.
        \State Apply $u_t$ and observe $y_{i,t}$ ($i=1,2,\dots,p$).
        \State Compute $i^{\star}$ from (\ref{eqn:i_star}). \Comment{next constraint to ride.}
        \State Compute $g_t$ from (\ref{eqn:g}).
        \State Update $\theta_{t+1}$ from (\ref{eqn:gradient_step}).        
        \State $\mathcal{H}_{t+1}\gets \mathcal{H}_{t}\cup \lbrace (u_t,y_{i^{\star},t}) \rbrace$.
        \State $t\gets t+1$
      \EndWhile
\end{algorithmic}\label{Alg:BangRide}
\end{algorithm}
\subsection{Convergence analysis}
The theoretical performance of the proposed online algorithm are evaluated via their total regret
\begin{equation}\label{eqn:regret}
    \mathcal{R}_{t_f}=\sum_{t=0}^{t_f}J_t(\theta_t)- J_t(\theta^{\star}_t).
\end{equation}
In addition to the usual assumptions of convexity and boundedness, we assume that the sequence of optimizers in (\ref{eqn:online_opt}) is convergent. This assumption, which is also in place in the SPSA algorithm~\cite{spall}, is required for regret analysis. Because, unlike the usual setting for online convex optimization algorithms, the regret defined in (\ref{eqn:regret}) is not restricted to statically optimal solutions (constant sequences)~\cite{Zink2003}.

\begin{theorem}\label{thm:regret}
Assume $\Theta\subseteq \mathbb{R}^m$ is a bounded set, $\theta_0\in\Theta$, $J_t(\theta)$ is convex in $\theta$, $\varepsilon_t=\Vert\theta^{\star}_{t+1}-\theta^{\star}_{t}\Vert=O(t^{-\mu_2})$, and that Assumption~\ref{ass:h} holds. Let the step size be
\begin{equation}\label{eqn:eqn_step_size}
    \alpha_t=\left\lbrace
    \begin{array}{ll}
        1, & t=0 \\
        t^{-\mu_1}, & t>0
    \end{array}\right.
\end{equation}
where $\mu_1\in(0,1)$ and $\mu_1<\mu_2$. Let $\Vert g_t\Vert,\vert c_t\vert\leq G$ for some $G>0$ and assume that the sequence $c_{t+1}/\alpha_{t+1}-c_{t}/\alpha_{t}$ changes sign a finite (possibly zero) number of times. Then the regret of Algorithm~\ref{Alg:BangRide} satisfies
\begin{equation}\label{eqn:regret=O}
    \mathcal{R}_{t_f} =O({t_f}^{\mu^{\star}})
\end{equation}
where $\mu^{\star}=\max\lbrace \mu_1,1-\mu_1,1+\mu_1-\mu_2\rbrace$.
\end{theorem}
\begin{proof}
Since the objective function $J_t(\theta)$ is convex, 
\begin{equation}\label{eqn:thm.proof.stp1}
J_t(\theta_t)-J_t(\theta^{\star}_t)\leq \langle \nabla_{\theta}J_t(\theta_t),\theta_t-\theta^{\star}_t\rangle
\end{equation}
Since the projection operator is non-expansive, by using the relation (\ref{eqn:gradient_step}) we obtain
\begin{align}\label{eqn:thm.proof.stp1.5}
    \Vert\theta_{t+1}-\theta^{\star}_t \Vert^2 &=\Vert \operatorname{proj}_{\Theta}\lbrace\theta_t-\alpha_tg_t\rbrace -\theta^{\star}_t \Vert^2 \\
    &\leq \Vert \theta_t-\alpha_tg_t -\theta^{\star}_t \Vert^2 \nonumber\\
    &=\Vert \theta_t -\theta^{\star}_t \Vert^2 +\alpha_t^2\Vert g_t\Vert^2 -2 \alpha_t\langle g_t,\theta_t-\theta^{\star}_t\rangle \nonumber
\end{align}
If $c_t=0$ then $\nabla_{\theta}J_t(\theta_t)=c_tg_t=0$ and thereby,
$$
J_t(\theta^{\star}_t)=J_t(\theta_t)
$$
holds from (\ref{eqn:thm.proof.stp1}), that is the corresponding regret term in (\ref{eqn:regret}) equals zero. By Assumption~\ref{ass:h} and the fact that the scalars $\gamma_i$ are positive, $c_t> 0$ holds from (\ref{eqn:c_t}). Therefore to examine the worst-case regret, we assume $c_t\geq \epsilon>0$. This allows writing (\ref{eqn:thm.proof.stp1.5}) in terms of the gradient $g_t=\nabla_{\theta}J_t(\theta_t)/c_t$, 
\begin{align}\label{eqn:thm.proof.stp2}
    &\langle \nabla_\theta J_t(\theta_t),\theta_t-\theta^{\star}_t\rangle
    \leq \\
    &\frac{c_t}{2\alpha_t}
    \left(\Vert \theta_t -\theta^{\star}_t \Vert^2 -\Vert\theta_{t+1}-\theta^{\star}_t \Vert^2\right)
    +\frac{\alpha_t}{2c_t}\Vert \nabla_\theta J_t(\theta_t)\Vert^2 \nonumber
\end{align}
By combining inequalities (\ref{eqn:thm.proof.stp1}) and (\ref{eqn:thm.proof.stp2}), we obtain
\begin{align}\label{eqn:proofstep2.5}
    & J_t(\theta_t)-J_t(\theta^{\star}_t)\leq\\
    &\frac{c_t}{2\alpha_t}
    \left(\Vert \theta_t -\theta^{\star}_t \Vert^2 -\Vert\theta_{t+1}-\theta^{\star}_t \Vert^2\right)
    +\frac{\alpha_t}{2c_t}\Vert \nabla_\theta J_t(\theta_t)\Vert^2 \nonumber
\end{align}
Since
\begin{align*}
&\Vert \theta_{t+1}-\theta^{\star}_{t+1}\Vert^2
\leq \\
&\Vert \theta_{t+1}-\theta^{\star}_{t}\Vert^2 +
\Vert \theta^{\star}_{t+1}-\theta^{\star}_{t}\Vert^2+
2\Vert \theta^{\star}_{t+1}-\theta^{\star}_{t}\Vert \Vert \theta_{t+1}-\theta^{\star}_t\Vert
\end{align*}
one has
\begin{align*}
&-\Vert \theta_{t+1}-\theta^{\star}_{t}\Vert^2 \leq \\
&\Vert \theta^{\star}_{t+1}-\theta^{\star}_{t}\Vert^2+
2\Vert \theta^{\star}_{t+1}-\theta^{\star}_{t}\Vert \Vert \theta_{t+1}-\theta^{\star}_t\Vert-\Vert \theta_{t+1}-\theta^{\star}_{t+1}\Vert^2 
\end{align*}
which, when applied to (\ref{eqn:proofstep2.5}), yields
\begin{align}\label{eqn:thm.proof.stp3}
&J_t(\theta_t)-J_t(\theta^{\star}_t)\leq \\
&\frac{\alpha_t}{2c_t}\Vert \nabla_\theta J_t(\theta_t)\Vert^2 +\nonumber\\
&\frac{c_t}{2\alpha_t}
    \left(\Vert \theta_t -\theta^{\star}_t \Vert^2 -
    \Vert \theta_{t+1}-\theta^{\star}_{t+1}\Vert^2\right) +\nonumber\\
&\frac{c_t}{2\alpha_t}\left(\Vert \theta^{\star}_{t+1}-\theta^{\star}_{t}\Vert^2+
2\Vert \theta^{\star}_{t+1}-\theta^{\star}_{t}\Vert \Vert \theta_{t+1}-\theta^{\star}_t\Vert\right) \nonumber
\end{align}
where we have used $c_t>0$. Summing both sides of (\ref{eqn:thm.proof.stp3}) for $t=0,1,\dots,t_f$ and using the upper bound $\Vert\nabla_{\theta}J_t(\theta_t)\Vert=c_t\Vert g_t\Vert\leq G^2$ therein results
\begin{align}\label{eqn:thm.proof.stp4}
&\sum_{t=0}^{t_f}J_t(\theta_t)-\sum_{t=0}^{t_f} J_t(\theta^{\star}_t)\leq  G^4\sum_{t=0}^{t_f}\frac{\alpha_t}{2c_t}+\\
&\frac{c_0}{2\alpha_0}\Vert \theta_0-\theta^{\star}_0\Vert^2 +
\sum_{t=0}^{t_f} \left( \frac{c_{t+1}}{2\alpha_{t+1}}-\frac{c_{t}}{2\alpha_{t}}\right)\Vert \theta_{t+1}-\theta^{\star}_{t+1}\Vert^2 +\nonumber \\
& \sum_{t=0}^{t_f}\frac{c_t}{2\alpha_t}\Vert \theta^{\star}_{t+1}-\theta^{\star}_{t}\Vert\left(\Vert \theta^{\star}_{t+1}-\theta^{\star}_{t}\Vert+
2 \Vert \theta_{t+1}-\theta^{\star}_t\Vert\right). \nonumber
\end{align}
Since
$$
\Vert \theta_{t}-\theta^{\star}_{t}\Vert,
\Vert \theta^{\star}_{t+1}-\theta^{\star}_{t}\Vert,
\Vert \theta_{t+1}-\theta^{\star}_{t}\Vert\leq \Vert\Theta\Vert
$$
where $\Vert\Theta\Vert<\infty$ is the diameter of the bounded set $\Theta$, it holds from (\ref{eqn:thm.proof.stp4}) that
\begin{align}\label{eqn:thm.proof.stp4.5}
&\sum_{t=0}^{t_f}J_t(\theta_t)-\sum_{t=0}^{t_f} J_t(\theta^{\star}_t) \nonumber\leq  G^4\sum_{t=0}^{t_f}\frac{\alpha_t}{2c_t}+\\
&\Vert \Theta\Vert^2 \frac{c_0}{2\alpha_0} +
\Vert \Theta\Vert^2\sum_{t\in\mathcal{T}} \left( \frac{c_{t+1}}{2\alpha_{t+1}}-\frac{c_{t}}{2\alpha_{t}}\right) +\nonumber \\
& 3\Vert\Theta\Vert\sum_{t=0}^{t_f}\frac{c_t}{2\alpha_t}\Vert \theta^{\star}_{t+1}-\theta^{\star}_{t}\Vert \nonumber\\
&\leq\frac{G^4}{2}\sum_{t=0}^{t_f}\frac{\alpha_t}{c_t}
+\Vert \Theta\Vert^2 \frac{c_0}{2\alpha_0} +\frac{\Vert \Theta\Vert^2}{2}\sum_{t\in\mathcal{S}}\frac{c_{t+1}}{\alpha_{t+1}} \nonumber\\
&+\frac{3\Vert\Theta\Vert}{2}\sum_{t=0}^{t_f}\frac{c_t}{\alpha_t}\varepsilon_t
\end{align}
where $\mathcal{T}=\lbrace t\in[0,t_f]\vert c_{t+1}/\alpha_{t+1}> c_{t}/\alpha_{t}\rbrace$ and $\mathcal{S}$ is a finite subset of $\mathcal{T}$.
With a large enough $t_f$, there exist some $\tau \geq 0$ and $M_0>0$ such that $\varepsilon_t\leq M_0t^{-\mu_2}$ holds for $\tau <t\leq t_f$. Therefore, by using the step size (\ref{eqn:eqn_step_size}) in (\ref{eqn:thm.proof.stp4.5}) we obtain
\begin{align*}
&\sum_{t=0}^{t_f}J_t(\theta_t)-\sum_{t=0}^{t_f} J_t(\theta^{\star}_t)\leq
\frac{G^4}{2\epsilon}\left(1+\sum_{t=1}^{t_f} t^{-\mu_1}\right)\\
&+\Vert\Theta\Vert^2\frac{G}{2}\bigl( 1+\vert\mathcal{S}\vert(t_f+1)^{\mu_1} \bigr)\nonumber\\
&+\Vert\Theta\Vert\frac{3G}{2}
\left(
\sum_{t=0}^{\tau} \frac{\varepsilon_t}{\alpha_t}
+M_0 \sum_{t=\tau+1}^{t_f} t^{\mu_1-\mu_2}
\right)
\end{align*}
where $\vert\mathcal{S}\vert<\infty$ is the cardinality of the set $\mathcal{S}$. Therefore when $\mu_2\neq 1+\mu_1$, the regret of Algorithm~\ref{alg} satisfies
\begin{align}\label{eqn:thm.proof.stp5}
&\mathcal{R}_{t_f}\leq \frac{G^4}{2\epsilon}
\left(2+\int_1^{t_f} t^{-\mu_1}dt \right)+\Vert\Theta\Vert^2\frac{G}{2}\bigl( 1+\vert\mathcal{S}\vert(t_f+1)^{\mu_1} \bigr)\nonumber\\
&+\Vert\Theta\Vert\frac{3G}{2}
\left(\sum_{t=0}^{\tau} \frac{\varepsilon_t}{\alpha_t} + M_0 (\tau+1)^{\mu_1-\mu_2}\right)\nonumber\\
&+M_0\Vert\Theta\Vert\frac{3G}{2}\int_{\tau+1}^{t_f} t^{\mu_1-\mu_2}dt \nonumber\\
&\leq M_1 + M_2 t_f^{1-\mu_1} + M_3{t_f}^{\mu_1}
+M_4 t_f^{1+\mu_1-\mu_2}
\end{align}
for some positive constants $M_1$, $M_2$, $M_3$ and $M_4$. When $\mu_2= 1+\mu_1$, on the other hand, a similar calculation yields
\begin{equation}\label{eqn:thm.proof.stp6}
\mathcal{R}_{t_f}\leq  M_1 + M_2 t_f^{1-\mu_1} + M_3{t_f}^{\mu_1}
+M_4 \log(t_f)
\end{equation}
In either case, both inequalities (\ref{eqn:thm.proof.stp5}) and (\ref{eqn:thm.proof.stp6}) imply
$$
\mathcal{R}_{t_f}\leq M{t_f}^{\mu^{\star}}
$$
for a large enugh $t_f>0$ and some $M>0$.
\end{proof}

Theorem~\ref{thm:regret} asserts that the regret of Algorithm~\ref{Alg:BangRide} is sublinear when $\mu^{\star}<1$. The minimum order of regret provable by Theorem~\ref{thm:regret} is
\begin{align*}
    \mu^{\star} &= \underset{\mu_1}{\min} \max \{ 1-\mu_1, \mu_1, 1+\mu_1-\mu_2\} \\
    &= \begin{cases}
        1-\mu_2/2 & \mbox{ if } \mu_2\in (0,1)\\
        1/2 &\mbox{ if } \mu_2\geq 1
    \end{cases}
\end{align*}
attained for $\mu_1=\min\{1/2, \mu_2/2\}$.
If $\mu_2\geq 1$, this agrees with the order of regret $\mu^{\star}=1/2$ in~\cite{Zink2003}.
When the summand in (\ref{eqn:regret}) has a limit,
a sublinear regret implies convergence to an optimal solution, that is,
\begin{equation}\label{eqn:limJ=0}
    \lim_{t\to +\infty} J_t(\theta_t)-J_t(\theta^{\star}_t)=0.
\end{equation}



\section{NUMERICAL EXPERIMENTS}\label{sec:numeric}
In this section, we evaluate our model-free charging method on a few common battery fast-charging problems.

\subsection{SPMeT cell}\label{ex:SPMeT}
We first consider the nonlinear (discretized) single-particle battery model augmented by electrolyte and temperature states (SPMeT). The input is the applied current $u$, the output is the terminal voltage $V$ and state of charge (SOC) and the state vector $x$ is 
\begin{align}
x =
    \begin{bmatrix}
         c^{a, -} & c^{s, -} & c^{e, -} & c^{e, +} & T
    \end{bmatrix}
\end{align}
where $c^{a, -}$ and $c^{s, -}$ are the average and surface concentration of lithium in the negative particle, respectively, $c^{e, \pm}$ is the electrolyte lithium-ion concentration and $T$ the temperature. It is assumed that the diffusion in the positive particle is infinitely fast such that its concentration is obtained from the mass balance. The particle average and surface concentration are described by an equivalent hydraulic model as in~\cite{RCGKG2019}. Diffusion of lithium-ions in the electrolyte is modeled by a Padé approximated transfer function as in~\cite{YJYWZ2017}. The resulting model is

\begin{align}\label{eqn:SPMeT}
c^{a,-}_{t+1}&=c^{a,-}_t+\frac{\Delta t}{V_p^-F}u_t \\ \nonumber
c_{t+1}^{s,-} &= \frac{\mathcal{G} \Delta t}{\beta(1-\beta)\tau^-}c^{a,-}_t + \left(1-\frac{\mathcal{G}\Delta t}{\beta(1-\beta)\tau^-} \right)c_{t}^{s,-}\\\nonumber
&\quad+\frac{\Delta t}{V_p^-F(1-\beta)}u_t\\ \nonumber
c^{e, -}_{t+1} &= c^{e, -}_{t} + \frac{\Delta t D^- N^-_1}{L^-\varepsilon^-N_3^-} (c^{e, -}_0-c^{e, -}_{t})+\frac{\Delta t N^-_2}{V_e^-FN_3^-} u_t \\\nonumber 
c^{e, +}_{t+1} &= c^{e, +}_{t} + \frac{\Delta t D^+ N^+_1}{L^+\varepsilon^+N_3^+} (c^{e, +}_0-c^{e, +}_{t})+\frac{\Delta t N^+_2}{V_e^+FN_3^+} u_t \\\nonumber 
T_{t+1}&=T_t- a \Delta t (T_t-T_a)\\\nonumber
    &\quad+b \Delta t
    (\Delta\eta(u_t, x_t) + \Delta \Phi_e(u_t, x_t)) u_t\\\nonumber
V_t &= \Delta U({x}_t) + \Delta\eta(u_t, x_t) + \Delta \Phi_e(u_t, x_t)\\\nonumber
SOC_t &= \frac{c^{a,-}_t/c^-_{\rm max}-\theta_1}{\theta_2-\theta_1}\nonumber
\end{align}

where $\Delta t$ is the discrete time step, $V_p$ aggregated particle volume, $F$ Faraday constant, $\tau$ solid diffusion time constant, $L$ electrode thickness, $\varepsilon$ porosity, $D^{\pm}$ diffusion coefficient, $V_e$ electrolyte volume, $c^-_{\rm max}$ maximum concentration and $\theta$ stoichiometric points defining SOC. Refer to \cite{RCGKG2019,YJYWZ2017} for computation of constants $\mathcal{G}$, $\beta$, $N_1^{\pm}$, $N_2^{\pm}$, $N_3^{\pm}$. 
Lastly, $\Delta U$, $\Delta \eta$ and $\Delta \Phi_e $ are non-linear functions that describe differences in open circuit potentials, overpotentials, and electrolyte potentials between the positive and negative electrodes. We consider a charging horizon of $t_f=3000$ with the initial conditions
\begin{align*}
&c^{a,-}_0=c^{s,-}_0=0.1 c_{max}^-,\\
&T_0=25,\\
&c^{e, -}_0 =c^{e, +}_0=1200,
\end{align*}
while enforcing the following bounds on the charging current and the terminal voltage
\begin{align}\label{eqn:constr_SPMet}
    u_t&\leq 2Q=56.3739, \nonumber\\
    V_t&\leq 4.2 
\end{align}
where $Q$ is the capacity. The objective is to maximize SOC which is equivalent to
$$\sum_{t=0}^{t_f} c^{a,-}_t,$$
as $c^{a,-}$ has a monotone relation with SOC according to (\ref{eqn:SPMeT}). One may implement the ideal bang-ride protocol (\ref{eqn:selector}), by solving the constraint equations (\ref{eqn:constraint_equation}) analytically using the model (\ref{eqn:SPMeT}). However, this is not possible, when the battery model is not available. In this case, one can use Algorithm~\ref{alg} to obtain the bang-ride charging current, by assuming a simple PI controller structure as follows
\begin{equation}\label{eqn:pi}
u_t=C(\mathcal{H}_t,\theta_t)=[\theta_t]_1 e_{i^{\star}(t-1),t-1} + [\theta_t]_2\sum_{k=0}^{t-1} e_{i^{\star}(k),k}
\end{equation}
where $[\theta_t]_1=K_p$ and $[\theta_t]_2=K_i$ are the PI gains. The error weights $\gamma_i$ in (\ref{eqn:cnstr_error}) are equally chosen, \emph{i.e.} $\Gamma=\textnormal{diag}(\gamma)=I$. In Figure~\ref{fig:SPMeT}, the resulting charging current, terminal voltage, temperature, and the state-of-charge are compared with those obtained from an ideal bang-ride solution derived from the battery model analytically. The resulting charging profile is of CCCV form, as can be seen in Figure~\ref{fig:SPMeT}. 


\begin{figure}
     \centering
     \begin{subfigure}[b]{0.475\textwidth}
        \centering
        \includegraphics[width=1\linewidth]{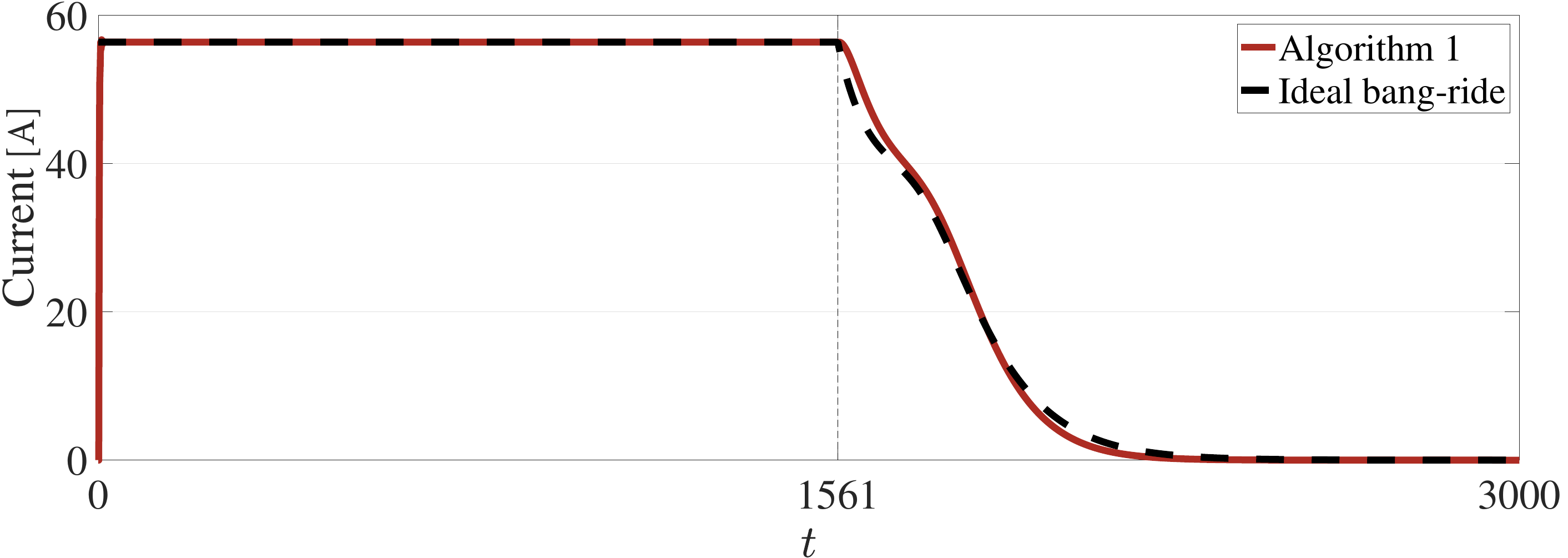}
	    \caption{Charging current over time in Example~\ref{ex:SPMeT}.}
         \label{fig:SPMet_i}
     \end{subfigure}
     \hfill
     \begin{subfigure}[b]{0.49\textwidth}
         \centering
        \includegraphics[width=1\linewidth]{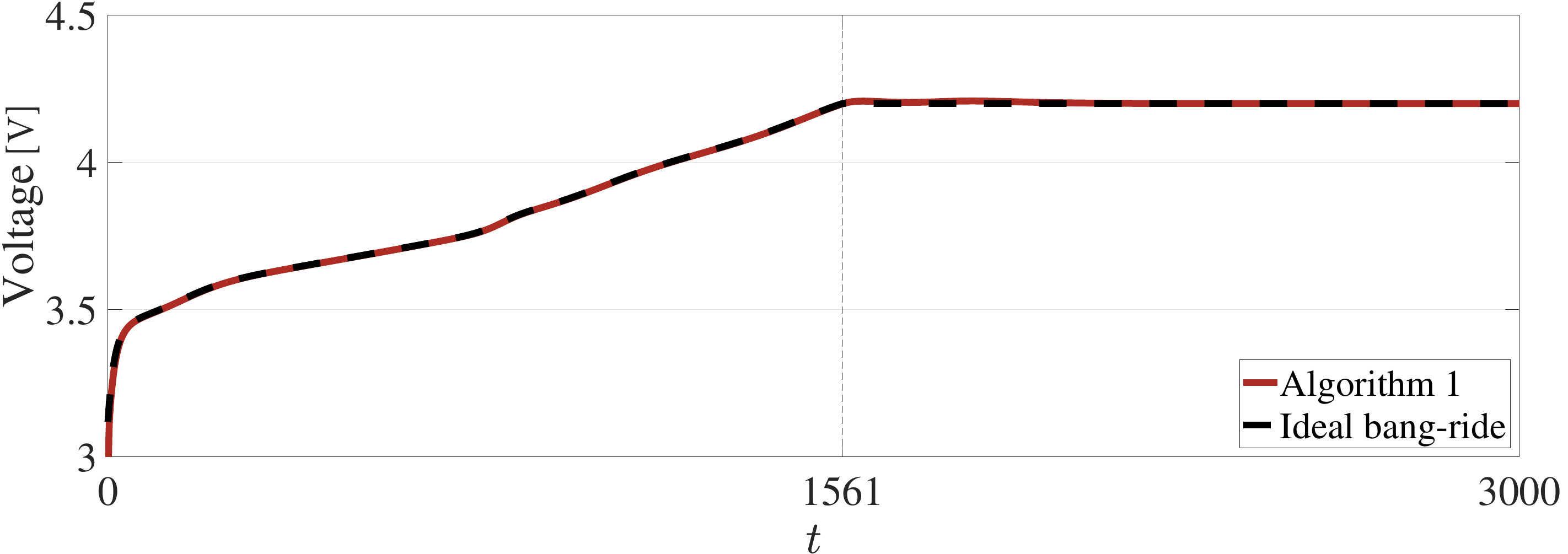}
	    \caption{Terminal voltage over time in Example~\ref{ex:SPMeT}.}
         \label{fig:SPMeT_v}
     \end{subfigure}
    \hfill
     \begin{subfigure}[b]{0.49\textwidth}
         \centering
        \includegraphics[width=1\linewidth]{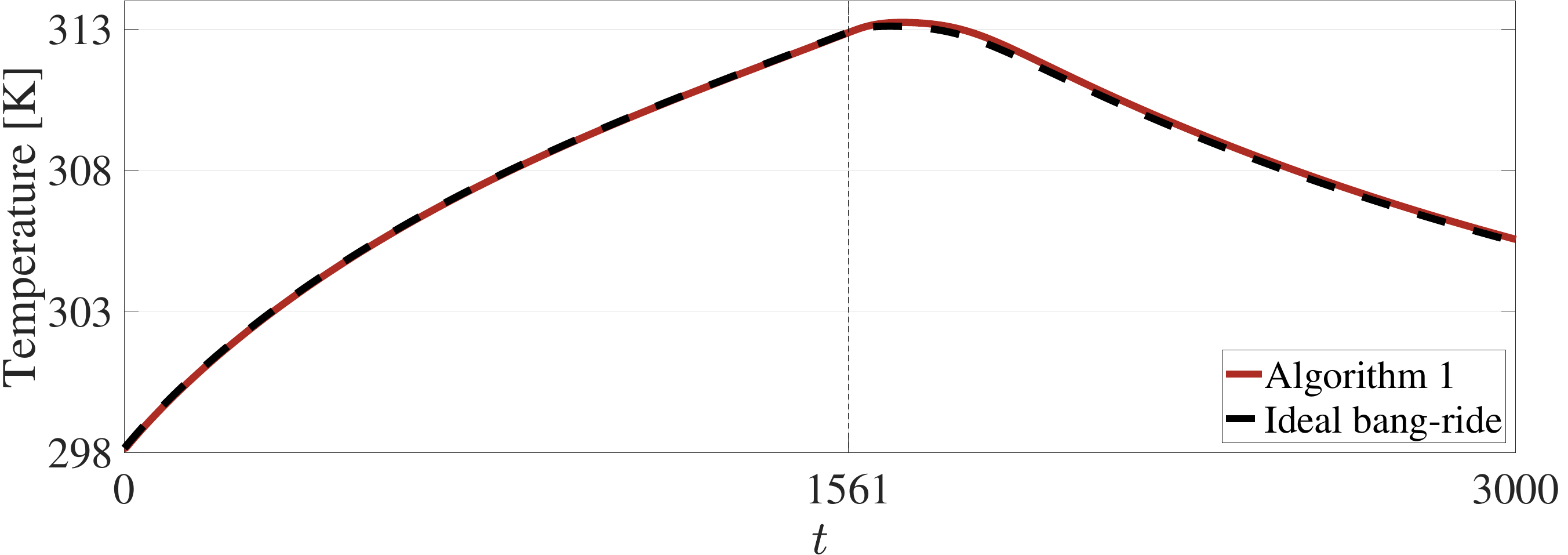}
	    \caption{Temperature over time in Example~\ref{ex:SPMeT}.}
         \label{fig:SPMeT_t}
     \end{subfigure}
    \hfill
     \begin{subfigure}[b]{0.49\textwidth}
         \centering
        \includegraphics[width=1\linewidth]{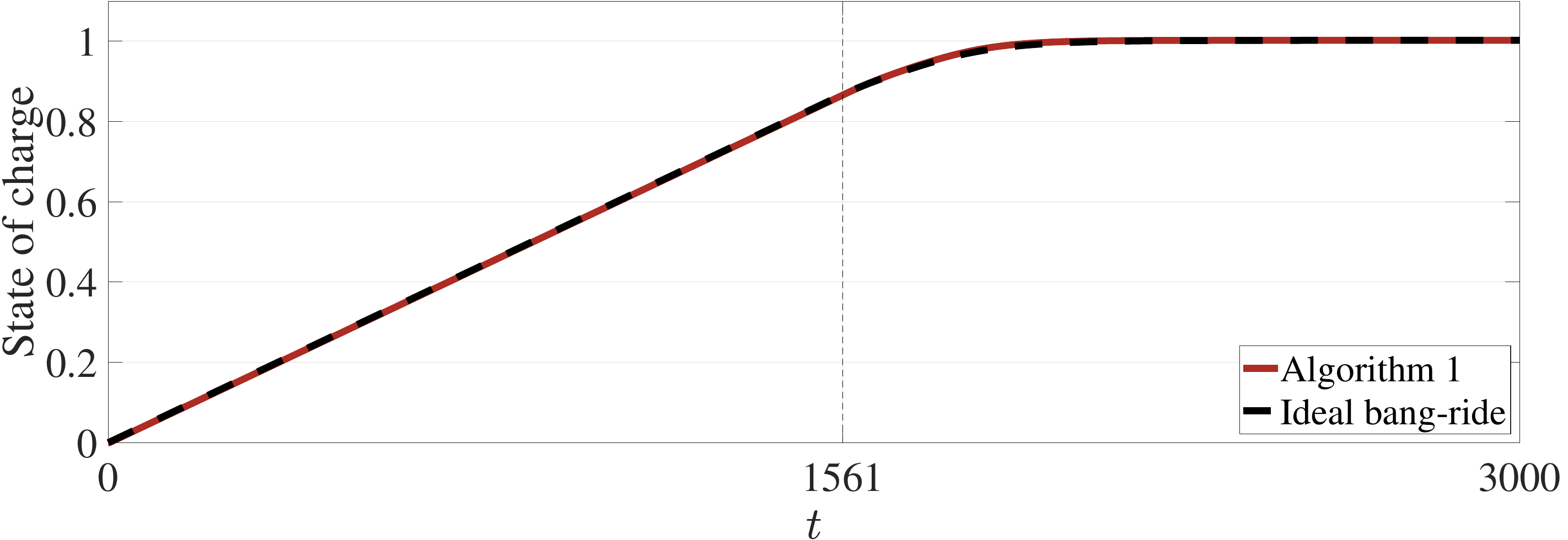}
	    \caption{State of charge over time in Example~\ref{ex:SPMeT}.}
         \label{fig:SPMeT_soc}
     \end{subfigure}
        \caption{The battery is charged with the maximum current until the terminal voltage constraint in (\ref{eqn:constr_SPMet}) is activated at $t=1561$, after which, the battery is charged with a constant voltage in Example~\ref{ex:SPMeT}.}
        \label{fig:SPMeT}
\end{figure}

\subsection{ECM cell}\label{ex:ECM1}
In the next example, we consider a continuous-time equivalent circuit battery model (ECM) with two time constants, augmented with the lumped thermal model~\cite{Hu2019}. This simple model is used to evaluate our algorithm against the bang-ride charging protocols computed from a large number of randomly generated perturbed models. We use the notation $u$ for the charging current, $T$ for the temperature, $SOC$ for the state of charge, and $v_1$, $v_2$ for the voltages across the (abstract) RC-links. This model is described by the following ordinary differential equations:
\begin{align}\label{eqn:ECM1}
\dot{v_1}(t)&=-\frac{1}{R_1C_1}v_1(t)+\frac{1}{C_1}u(t) \nonumber\\
\dot{v_2}(t)&=-\frac{1}{R_2C_2}v_2(t)+\frac{1}{C_2}u(t) \nonumber\\
\dot{SOC}(t)&=\frac{1}{Q}u(t) \nonumber\\
\dot{T}(t) &= - a (T(t) - T_a)+ b u(t)v_d(t),
\end{align}
where $v_d(t)=R_o u(t) + v_1(t)+v_2(t)$ is called the dynamic voltage. The battery model (\ref{eqn:ECM1}) can be written as a nonlinear monotone dynamic system of the form (\ref{eqn:system}) after discretization with an appropriate sampling time $\Delta t$. In this model, the measurable signals are the input current, terminal voltage, and the cell temperature, based on which we define the following outputs
\begin{align*}
    h_1(x_t,u_t)&=u_t,\\
    h_2(x_t,u_t)&=Kx_t+u_t,\\
    h_3(x_t,u_t)&=T_{t+1}-T_a \\
    &=[x_t]_4 (1-a\Delta t) + b\Delta t C x_t u_t+ b\Delta tR_0 u_t^2
\end{align*}
where $C,K \in\mathbb{R}^{1\times 4}$ are constant, $[x_t]_3$ is the state of charge, and $[x_t]_4=T_t-T_a$ is the temperature variation of the cell with respect to ambient at time $t$. Now consider the fast charging problem (\ref{eqn:optimization}) with $L(x_t,u_t)=[x_t]_3$. In addition to current and voltage constraints, we impose a constraint on the cell temperature. This additional constraint is known to help with limiting degradation and sustaining the negative electrode performance~\cite{Palac2018}. These constraints can be written as $y_{i,t}\leq \bar{y}_i$ where
\begin{align*}
\bar{y}_1&= u_{\rm max}, \\
\bar{y}_2&=v_{\rm max}, \\ 
\bar{y}_3&=T_{\rm max}-T_{a}
\end{align*}
We consider a simple PI controller as (\ref{eqn:pi}) in Algorithm~\ref{alg}. The charging current, dynamic voltage, and the cell temperature following the model-free charging algorithm provided in this paper are plotted in Figure~\ref{fig:ECM1} (solid red and blue), along with those of the ideal bang-ride charging protocol (dashed black). The charging current from Algorithm~\ref{alg} is computed online based on the input-output data, whereas the bang-ride charging currents are obtained analytically, by solving the constraint equations (\ref{eqn:constraint_equation}) and selecting the minimum state-feedback control signals (\ref{eqn:selector}). We also computed the ideal bang-ride currents for $1000$ different random ECMs whose parameters differed from the true values by a maximum of $10\%$. The resulting currents, dynamic voltages, and temperatures of the true battery under these protocols are plotted in solid gray. The bang-ride protocols computed from these perturbed models can violate the constraints or lead to sub-optimal operation. In contrast, Algorithm~\ref{alg} tracks the ideal bang-ride protocol without having access to the (even inaccurate) model of the battery (see Figure~\ref{fig:ECM1}).

As can be seen in Figure~\ref{fig:ECM1}, when the error signals (\ref{eqn:cnstr_error}) are not scaled properly ($\Gamma=I$) the model-free approach is slow in tracking the ideal bang-ride charging protocol once the temperature constraint is switched on. The reason is that the temperature dynamics are substantially slower than the voltage or the current dynamics and when it is switched on, the controller gains $\theta$ struggle to adapt to the abrupt change. As can be seen in Figure~\ref{fig:scheme}, this problem is easily treated by scaling the measured errors so that they are roughly in the same order of magnitude as follows
\begin{equation}\label{eqn:gamma=500}
    \Gamma=\begin{bmatrix}
    1&0&0\\
    0&1&0\\
    0&0&500
\end{bmatrix}
\end{equation}

\begin{figure}
     \centering
     \begin{subfigure}[b]{0.49\textwidth}
        \centering
        \includegraphics[width=1\linewidth]
        {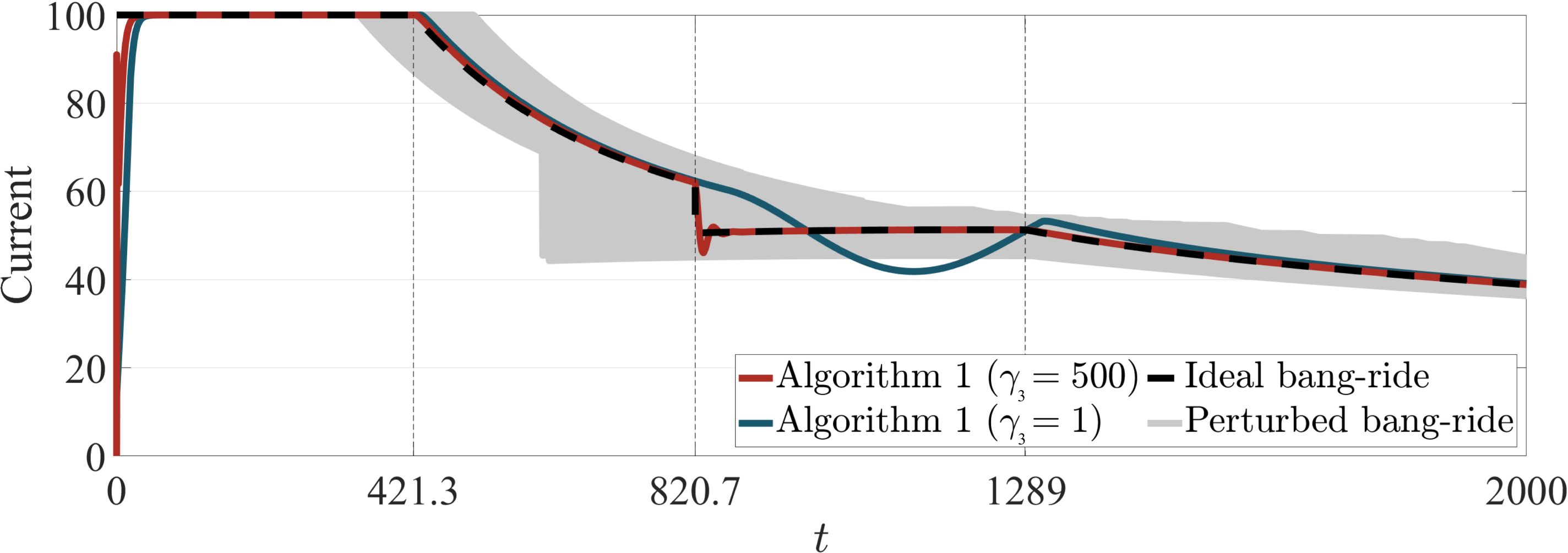}
	    \caption{Charging current over time in Example~\ref{ex:ECM1}.}
         \label{fig:ECM_i}
     \end{subfigure}
     \hfill
     \begin{subfigure}[b]{0.49\textwidth}
         \centering
        \includegraphics[width=1\linewidth]
        {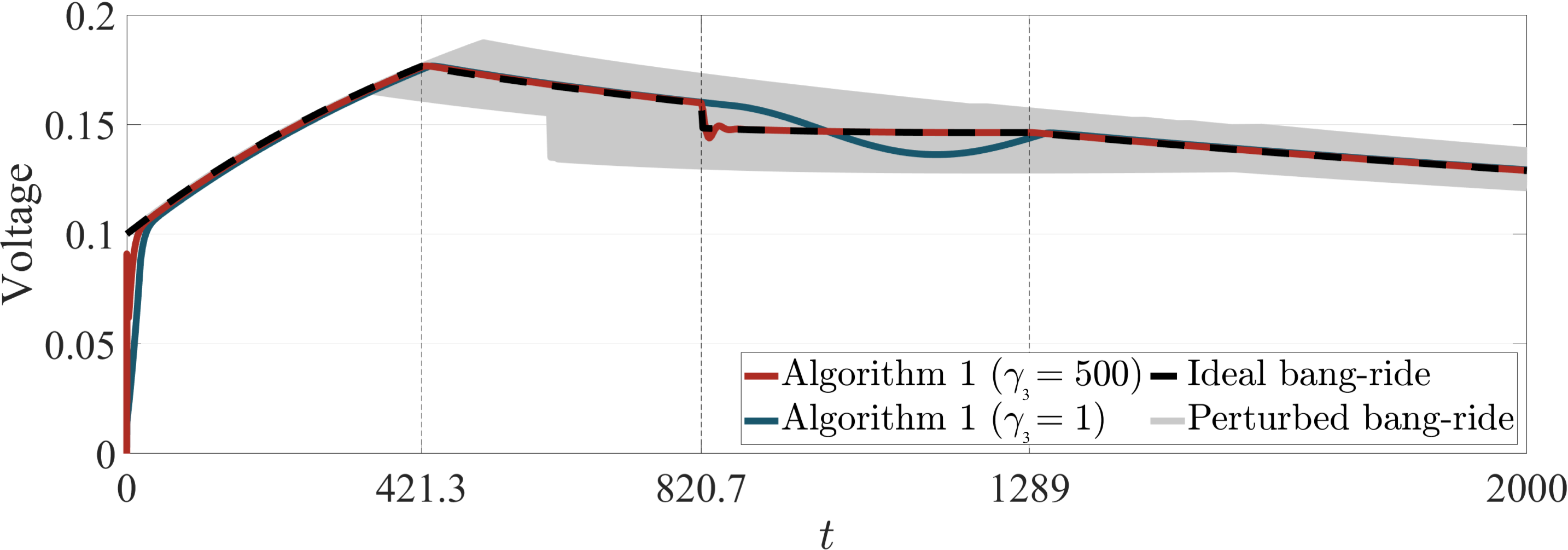}
	    \caption{Dynamic voltage over time in Example~\ref{ex:ECM1}.}
         \label{fig:ECM_v}
     \end{subfigure}
    \hfill
    \begin{subfigure}[b]{0.49\textwidth}
         \centering
        \includegraphics[width=1\linewidth]
        {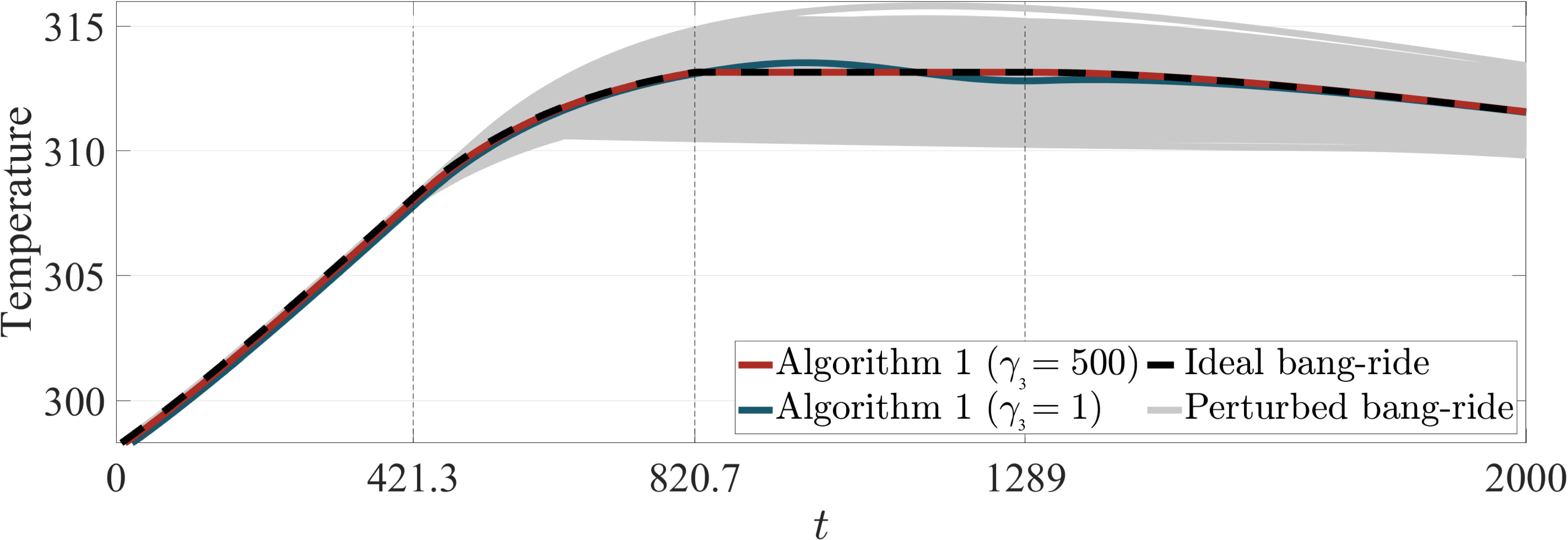}
	    \caption{Temperature over time in Example~\ref{ex:ECM1}.}
    \label{fig:ECM_t}
    \end{subfigure}
        \caption{The active constraint changes three times in Example~\ref{ex:ECM1}. In the beginning, the current constraint ($y_{1,t}= u_{\rm max}$) is active to charge with the maximum current until $t=421.3$, when the current constraint is switched to the voltage constraint. The voltage constraint is active until $t=820.7$ before the active constraint is switched to the temperature constraint. Finally, the active constraint is switched back to the voltage constraint at $t=1289$ which remains active until the end of the experiment.}
        \label{fig:ECM1}
\end{figure}

\subsection{ECM pack}\label{ex:ECM2}
As the final experiment, we consider a pack of 100 battery cells each described by an augmented ECM (\ref{eqn:ECM1}), but with modified temperature dynamics to account for the heat exchanged between adjacent cells in the pack as follows~\cite{Pozzi2022}
\begin{align}\label{eqn:temp_transfer}
    \dot{T}_i(t) &= - a (T_i(t) - T_a)\\
    &+ b u(t)\left( R_o u(t) + v_{i,1}(t)+v_{i,2}(t)\right) \nonumber\\
&+k_1(T_{i-1}(t)-T_{i}(t))+k_2(T_{i+1}(t)-T_{i}(t)) \nonumber
\end{align}
where $k_1$ and $k_2$ are some positive constants and $T_i$ is the temperature of the $i$-th cell. For convenience, when $i=1$ and $i=100$, we take $T_0(t)=T_{100}(t)$ and $T_{101}(t)=T_1(t)$ in (\ref{eqn:temp_transfer}) respectively. The parameters of the RC links in the equivalent circuit models are slightly varied between the cells. The cells are assumed to be connected in series, so by Kirchhoff laws, the same input current $u(t)$ goes through all of them and the overall pack terminal voltage equals the sum of the terminal voltages of all cells as follows
$$
V(t)=V_1(t)+V_2(t)+\cdots +V_{100}(t)
$$
where $V_i(t)=v_{i,1}(t)+v_{i,2}(t)+R_{0}u(t)+v_{\rm ocv}$ is the terminal voltage of the $i$-th cell and $v_{\rm ocv}$ is the cell open-circuit voltage. This model can be discretized with an appropriate sampling time $\Delta t$ to obtain the battery pack model (\ref{eqn:system}), with the measurable outputs defined as follows
\begin{align*}
    h_i(x_t,u_t)&=u_t,& i=1\\
    h_i(x_t,u_t)&=Kx_{i,t}+u_t,&  2\leq i\leq 101 \\
    h_i(x_t,u_t)&=[x_{i,t}]_4 (1-a\Delta t)&\\
    &+b\Delta t C x_{i,t} u_t + b\Delta tR_0 u_t^2 &\\
    &+\Delta t k_1([x_{i-1,t}]_4-[x_{i,t}]_4)&\\
    &+\Delta t k_2 ([x_{i+1,t}]_4-[x_{i,t}]_4), & 102\leq i\leq 201 \\
    h_i(x_t,u_t)&=h_j(x_t,u_t)-h_k(x_t,u_t), & 202\leq i\leq 10201
\end{align*}
where the last $10000$ output functions are constructed by the differences between the temperature variations of any two cells within the pack, \emph{i.e.}, $j,k\in\lbrace 102, 103,\dots, 201\rbrace$.
We consider the fast charging problem (\ref{eqn:optimization}) with the total state of charge $$L(x_t,u_t)=\sum_{i=1}^{100} [x_{i,t}]_3.$$ In addition to the constraints on every cell voltage and temperature, we also impose a constraint on the differences between the cells' temperatures. The purpose of this constraint is to achieve a uniform degradation of the cells. All the constraints can be written in the form $y_{i,t}\leq \bar{y}_i$, where
\begin{align*}
\bar{y}_i&= u_{\rm max},& i&=1\\
\bar{y}_i&=v_{\rm max}, & i&=2,3,\dots,101\\
\bar{y}_i&=T_{\rm max}-T_{a}, & i&=102,103,\dots,201 \\
\bar{y}_i&= \Delta T_{\rm max}, & i&=202,203,\dots,10201 
\end{align*}
and $\Delta T_{\rm max}=5$ is the maximum allowable temperature difference between any two cells in the pack. The error weights in (\ref{eqn:cnstr_error}) are chosen as
\begin{align*}
\gamma_i&= 1, & i&=1,3,\dots,101\\
\gamma_i&=500, & i&=102,103,\dots,10201. 
\end{align*}
Similar to Examples~\ref{ex:SPMeT} and \ref{ex:ECM1}, we consider a simple PI controller structure as (\ref{eqn:pi}) in Algorithm~\ref{alg}. The resulting charging current, the overall pack voltage, and the maximum and minimum temperatures among the cells in the pack are shown in Figure~\ref{fig:ECM2}.

\begin{figure}
     \centering
     \begin{subfigure}[b]{0.475\textwidth}
        \centering
        \includegraphics[width=1\linewidth]{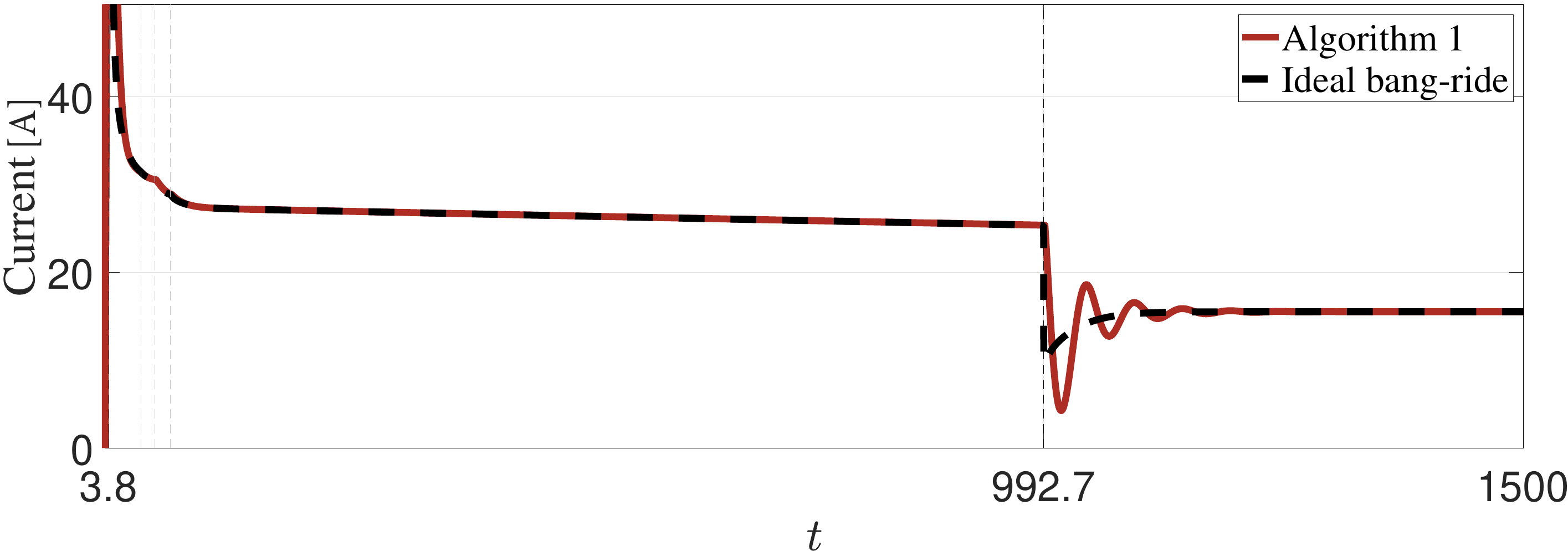}
	    \caption{Charging current over time in Example~\ref{ex:ECM2}.}
         \label{fig:ECM2_i}
     \end{subfigure}
     \hfill
     \begin{subfigure}[b]{0.49\textwidth}
         \centering
        \includegraphics[width=1\linewidth]{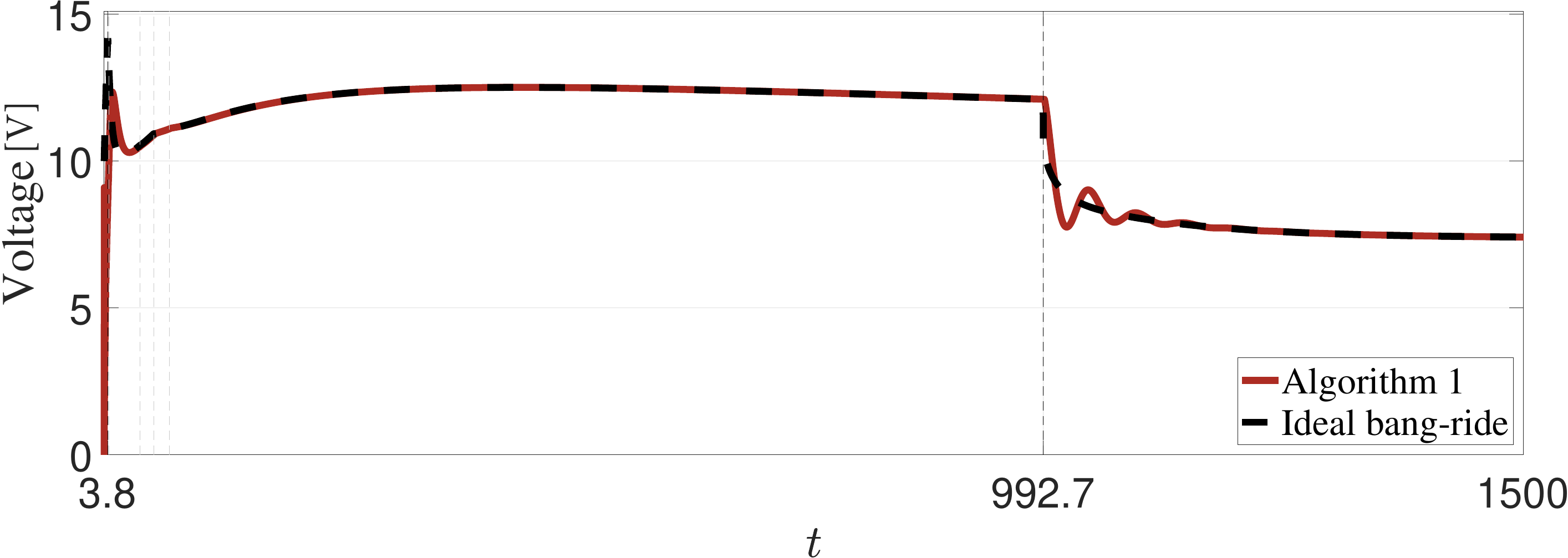}
	    \caption{Terminal voltage over time in Example~\ref{ex:ECM2}.}
         \label{fig:ECM2_v}
     \end{subfigure}
    \hfill
    \begin{subfigure}[b]{0.49\textwidth}
         \centering
        \includegraphics[width=1\linewidth]{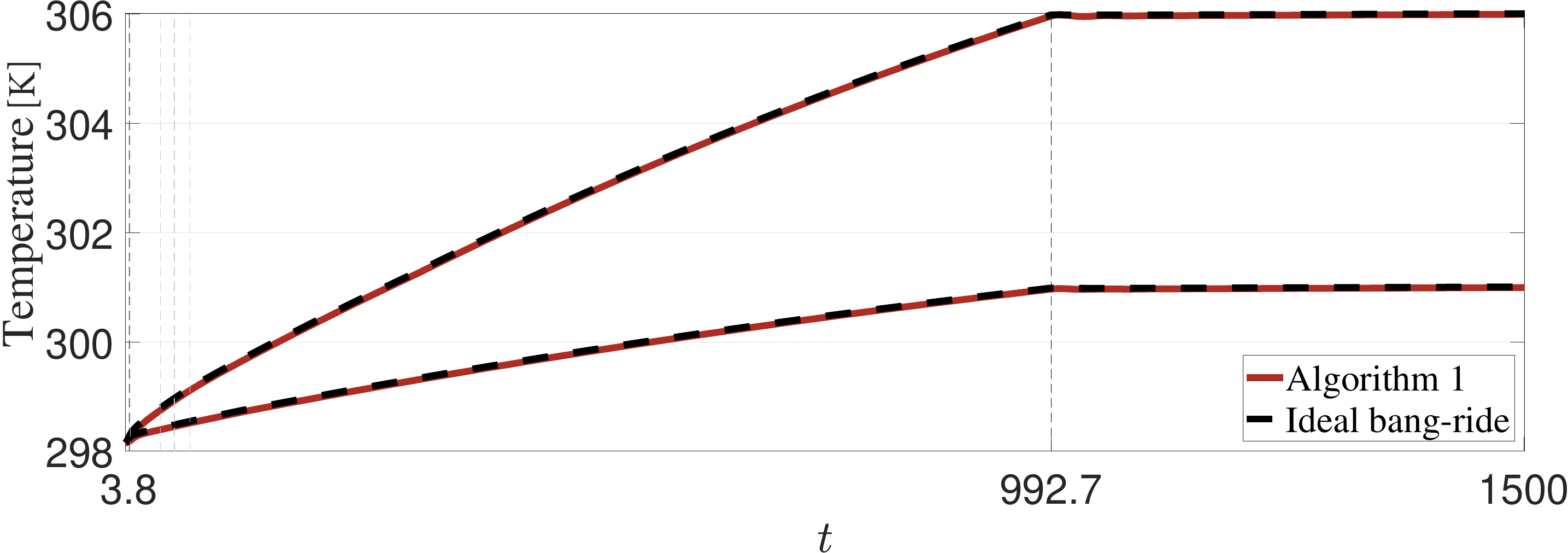}
	    \caption{Temperature over time in Example~\ref{ex:ECM2}.}
    \label{fig:ECM2_t}
    \end{subfigure}
        \caption{Under the ideal bang-ride charging approach in Example~\ref{ex:ECM2}, the battery is charged with the maximum current $u_t=u_{\rm max}$ until $t=3.8$ when a cell voltage constraint is activated. After that, the active voltage constraint is exchanged several times among different cells in the pack until $t=992.7$ when it is handed over to a temperature difference constraint. The last activated constraint maintains a difference of $\Delta T_{\rm max}$ between the cells' temperatures in the pack.}
        \label{fig:ECM2}
\end{figure}

\section{CONCLUSION}\label{sec:conclusion}
We considered a general battery fast charging problem with various constraints in this paper. We proposed a model-free algorithm that converges to the ideal bang-ride protocol, by observing the charging data without accessing the battery model or its states. The numerical examples demonstrate that the proposed model-free approach can be more effective than first identifying an (imperfect) model for the battery and then computing its optimal charging profile.

In future work, it would be interesting to investigate if it is possible to ensure that the proposed algorithm stays in the feasible region at all times, satisfying the desired charging constraints. Since any overshoot in the data-driven closed-loop system (Figure (\ref{fig:scheme})) as $e_{i^{\star},t}<0$ results in a violation of the active constraint, developing data-driven output-feedback controllers that can track a reference with no overshoots would provide constraint satisfaction guarantees for our data-driven algorithm. However, this can potentially lead to overly cautious charging strategies, which would result in large regrets. This is the case when the closed-loop system in Figure~(\ref{fig:scheme}) has a slow response. Therefore, the model-free charging method proposed in this paper can benefit from developing a data-driven controller design that can achieve a non-overshooting closed-loop response with a fast decay rate for a system with unknown dynamics.

\printbibliography

\end{document}